\documentclass[5p,times]{elsarticle}

\usepackage{amssymb}
\usepackage{amsmath}
\usepackage{amsfonts}
\usepackage{xcolor}
\usepackage{algorithm}
\usepackage{algpseudocode}
\usepackage{amsthm}
\usepackage{subcaption}

\newtheorem{proposition}{Proposition}[section]

\newtheorem{assumption}{Assumption}[section]
\newtheorem{remark}{Remark}[section]

\journal{a Journal}

\begin{document}

\begin{frontmatter}



\title{Structured physics--guided neural networks for electromagnetic commutation applied to industrial linear motors\tnoteref{label10}}
\tnotetext[label10]{This work is supported by the NWO research project PGN Mechatronics, project number 17973.}

\author[labelTUE]{Max Bolderman\corref{cor1}}
\ead{m.bolderman@tue.nl}
\author[labelTUE]{Mircea Lazar}
\author[labelTUE,labelASML]{Hans Butler}

\cortext[cor1]{Corresponding author.}

\affiliation[labelTUE]{organization={Eindhoven University of Technology},addressline={Groene Loper 19},city={Eindhoven},postcode={5612 AP},country={The Netherlands}}
\affiliation[labelASML]{organization={ASML},
             addressline={De Run 6501},
             city={Veldhoven},
             postcode={5504 DR},
             country={The Netherlands}}

\begin{abstract}
Mechatronic systems are described by an interconnection of the electromagnetic part, i.e., a static position--dependent nonlinear relation between currents and forces, and the mechanical part, i.e., a dynamic relation from forces to position. 
Commutation inverts a model of the electromagnetic part of the system, and thereby removes the electromagnetic part from the position control problem.
Typical commutation algorithms rely on simplified models derived from physics--based knowledge, which do not take into account position dependent parasitic effects.
In turn, these commutation related model errors translate into position tracking errors, which limit the system performance.
Therefore, in this work, we develop a data--driven approach to commutation using physics--guided neural networks (PGNNs).
A novel PGNN model is proposed which structures neural networks (NNs) to learn specific motor dependent parasitic effects. 
The PGNN is used to identify a model of the electromagnetic part using force measurements, after which it is analytically inverted to obtain a PGNN--based commutation algorithm. 
Motivated by industrial applications, we develop an input transformation to deal with systems with fixed commutation, i.e., when the currents cannot be controlled. 
Real--life experiments on an industrial coreless linear motor (CLM) demonstrate a factor $10$ improvement in the commutation error in driving direction and a factor $4$ improvement in the position error with respect to classical commutation in terms of the mean--squared error (MSE). 
\end{abstract}




\end{frontmatter}


\section{Introduction}
\label{sec:Introduction}
The increasing demands on throughput and accuracy within high--precision motion systems require innovative solutions from a motion control perspective, see, e.g.,~\cite{Schmidt2014}. 
These motion systems are described as an interconnection between: $i)$ the electromagnetic part, i.e., the nonlinear position--dependent static relation between the coil currents and actuation forces, and $ii)$ the mechanical part, i.e., the dynamics from the actuation forces to the position output~\cite{Nguyen2018}. 
The commutation algorithm inverts the electromagnetic part of the system, and thereby removes the electromagnetic part from the position control problem. 
Therefore, commutation has similarities to feedback linearization~\cite{Samuelsson2005}, and is also known as control allocation~\cite{Johansen2013} or current waveform optimization~\cite{Rohrig2003,Moehle2015}.
Often, the system is overactuated, i.e., the number of independent coil currents exceeds the number of actuation forces. 
This overactuation is used in the commutation algorithm by, e.g., minimizing the energy consumption~\cite{Nguyen2016} or not to excite flexible body modes in the translator~\cite{Custers2019}.

Commutation relies on the inversion of a model that describes the electromagnetic part. 
Conventionally, this model is obtained from first principle knowledge, which, after inversion, yields the sinusoidal three--phase current waveforms~\cite{Gieras2011}. 
This has similarities to the $dq0$--reference frame often used in rotary systems~\cite{Overboom2015}. 
Unfortunately, this physics--based approach has limited accuracy, as it fails to describe position dependent parasitic effects present in real--life mechatronics systems. 
For example, it does not include manufacturing tolerances, such as mismatches in orientation and field intensity of the permanent magnets~\cite{Nguyen2018}.
As a result, we observe significant differences between the desired and generated forces, i.e., the commutation error. 
These errors reduce the reference tracking performance, and thereby degrade the overall system performance. 

With the aim to enhance commutation performance, several methods have been proposed.
In~\cite{Broens2023}, a feedback control loop is designed to steer an offset on the commutation phase. 
In~\cite{Strijbosch2019},~\cite{Aarnoudse2023}, an iterative learning approach has been developed that learns a current waveform iteratively, i.e., by learning from past trials. 
Alternatively, several methods have been developed that rely on identifying more accurately the electromagnetic part or the commutation waveforms. 
A position--dependent compensation scheme is developed in~\cite{Bascetta2010} to compensate for the force ripple caused by imperfect commutation.
A Fourier model is developed in~\cite{Rohrig2008} to shape the current waveforms in the commutation, and~\cite{Nguyen2018} develops a Fourier model for identification of the electromagnetic part.
In~\cite{Meer2022}, Gaussian Processes are used to parametrize the commutation waveforms for a switched reluctance motor. 

Motivated by the successful application of physics--guided neural networks (PGNNs) to feedforward control for mechatronic systems~\cite{Bolderman2021}, in this work, we develop a data--driven approach to commutation using PGNNs for a coreless linear motor (CLM).  
PGNNs effectively merge a known, physics--based model with a black box neural network (NN), and thereby retain the physical interpretability with improved accuracy compared to the stand--alone physics--based model~\cite{Bolderman2024}. 
This suits the commutation problem, where a simplified physics--based model used in classical commutation is known and can be embedded in a PGNN model architecture. 

The contributions of this paper are as follows.
Firstly, we develop a novel, invertible, PGNN model class for describing the electromagnetic part of a CLM by structuring neural networks (NNs) to learn motor dependent parasitic effects. The PGNN parameters are then identified using a data set that utilizes force measurements. By means of an optimal parameter selection, it is demonstrated that the PGNN better fits the data compared to the classical physics--based model.
Secondly, we compute a commutation law by analytically inverting the identified PGNN while minimizing energy consumption.
Then, motivated by common limitations in industrial applications, we design a transformation to control the commutation magnitude and phase instead of directly prescribing the currents.
Real--life experiments on an industrial CLM demonstrate a reduction of the commutation error in driving direction by a factor $10$ on a representative position reference in terms of the mean--squared error (MSE) with respect to classical commutation. 
Moreover, the MSE of the position tracking error reduces by a factor $4$.


\begin{figure}
\centering
\includegraphics[width=1.0\linewidth]{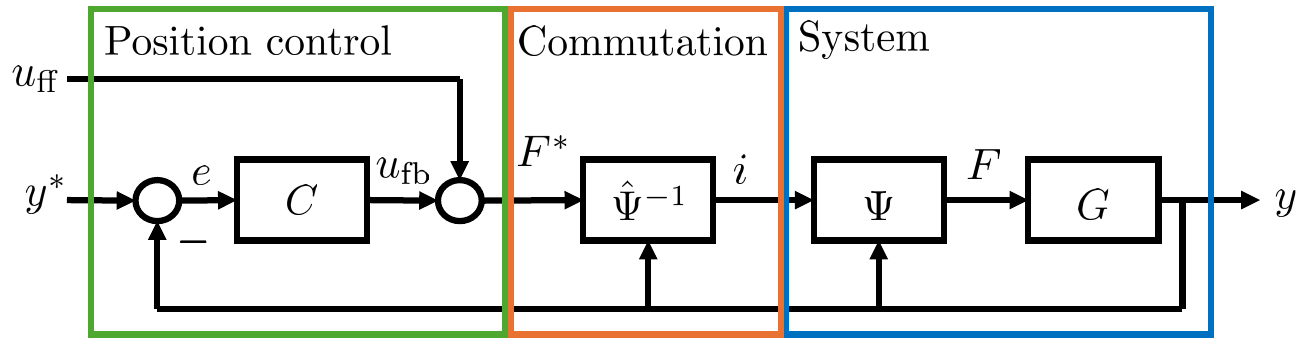}
\caption{Schematic overview of the closed--loop control architecture.}
\label{fig:Schematic_Overview}
\end{figure}

\section{Preliminaries}
\label{sec:Preliminaries}
\subsection{Notation}
We denote $\mathbb{R}$ and $\mathbb{Z}$ as the set of real and integer numbers, respectively. 
The positions at time $t \in \mathbb{R}_{>0}$ are denoted by $y(t) \in \mathbb{R}^{n_y}$, $n_y \in \mathbb{Z}_{>0}$, the forces acting on the translator are $F(t) \in \mathbb{R}^{n_F}$, $n_F \in \mathbb{Z}_{>0}$, and $i(t) \in \mathbb{R}^{3L}$ are the three--phase currents passing through each of the $L \in \mathbb{Z}_{>0}$ coil sets. 
A superscript $\cdot^*$ indicates a desired value, e.g., $y^*(t)$ is the desired position $y(t)$, and $F^*(t)$ are the desired forces $F(t)$.
A superscript $l \in [1, ..., L]$ refers to coil set $l$, e.g., $i^l(t) = [i_a^l(t), i_b^l(t), i_c^l(t)]^T$ are the three--phase currents passing through coil set $l$, and $F^l(t)$ are the forces generated by coil set $l$. 
A hat $\hat{\cdot}$ denotes a model or the output predicted by a model, e.g., $\hat{\Psi}$ is a model of electromagnetic part $\Psi$ and $\hat{F}(t)$ is a model--based prediction of the force $F(t)$.

\subsection{Closed--loop control architecture}
We start by introducing the standard closed--loop control architecture for high--precision mechatronic systems illustrated in Fig.~\ref{fig:Schematic_Overview}.
The mechanical part of the system, denoted by $G$, describes the dynamical behaviour of the system from force $F(t)$ to position $y(t)$.  
The mechanical dynamics are preceded by the electromagnetic part, which are a  static, position--dependent nonlinearity according to
\begin{equation}
\label{eq:Electromagnetics}
	F(t) = \Psi \big( i(t), y(t) \big),
\end{equation}
where $i(t) \in \mathbb{R}^{3 L}$ are the three--phase currents passing through the $L \in \mathbb{Z}_{>0}$ coil sets, and $\Psi : \mathbb{R}^{3 L} \times \mathbb{R}^{n_y} \rightarrow \mathbb{R}^{n_F}$ the electromagnetic part. 
The commutation aims to prescribe the currents $i(t)$ that generate a desired force $F^*(t)$ by inverting the electromagnetic part~\eqref{eq:Electromagnetics}, such that
\begin{equation}
\label{eq:Commutation}
	i (t) = \hat{\Psi}^{-1} \big( F^*(t), y(t) \big),
\end{equation}
where $\hat{\Psi}^{-1} : \mathbb{R}^{n_F} \times \mathbb{R}^{n_y} \rightarrow \mathbb{R}^{3L}$ is a model of the inverse of~$\Psi$.
The desired force is generated by the position control loop, which aims to minimize the tracking error $e(t) := y^*(t) - y(t)$ as
\begin{equation}
\label{eq:Position_Control}
	F^*(t) = u_{\textup{ff}}(t) + u_{\textup{fb}}(t),
\end{equation}
where $u_{\textup{ff}}(t) \in \mathbb{R}^{n_F}$ is the feedforward, and $u_{\textup{fb}}(t) \in \mathbb{R}^{n_F}$ the feedback input. 

Ideally, the currents prescribed by the commutation~\eqref{eq:Commutation} generate the desired forces when supplied to the system, such that
\begin{equation}
\label{eq:Commutation_Objective}
	F(t) = \Psi \Big( \hat{\Psi}^{-1} \big( F^*(t), y(t) \big), y(t) \Big) = F^*(t).
\end{equation}

\begin{figure}
\centering
\includegraphics[width=1.0\linewidth]{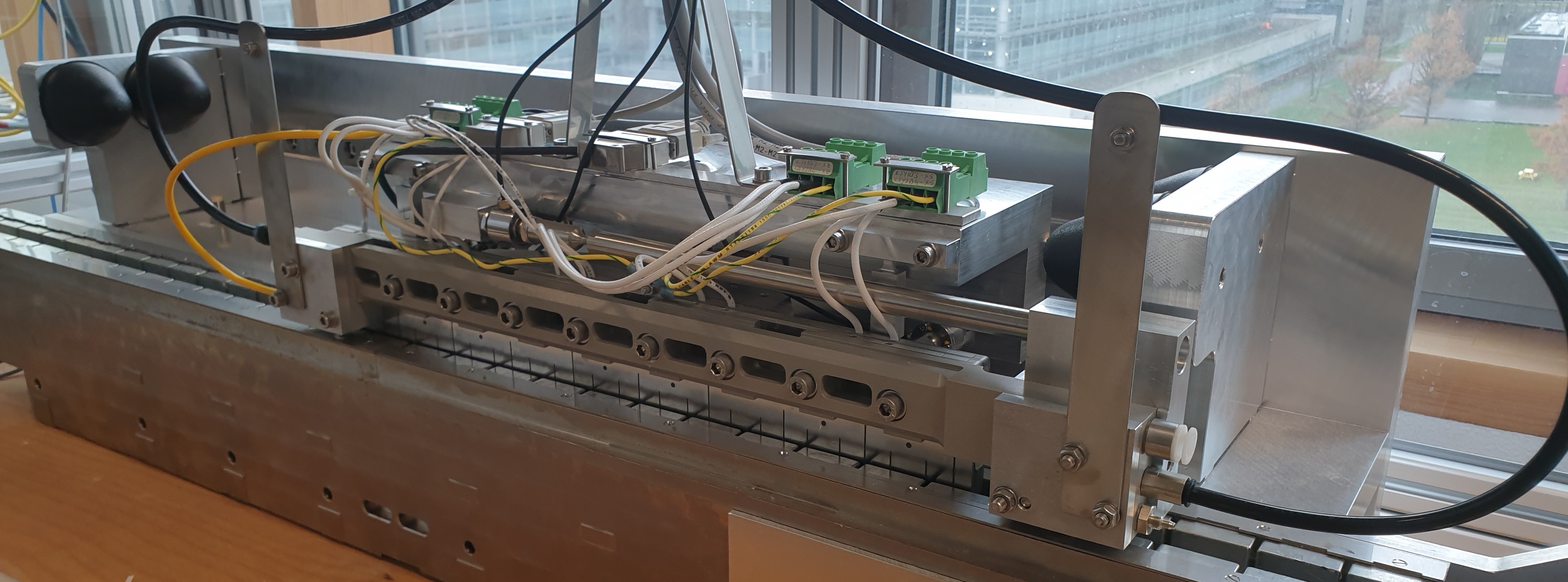}
\caption{Industrial coreless linear motor.}
\label{fig:CLM}
\end{figure}
\begin{figure}
\centering
\includegraphics[width=1.0\linewidth]{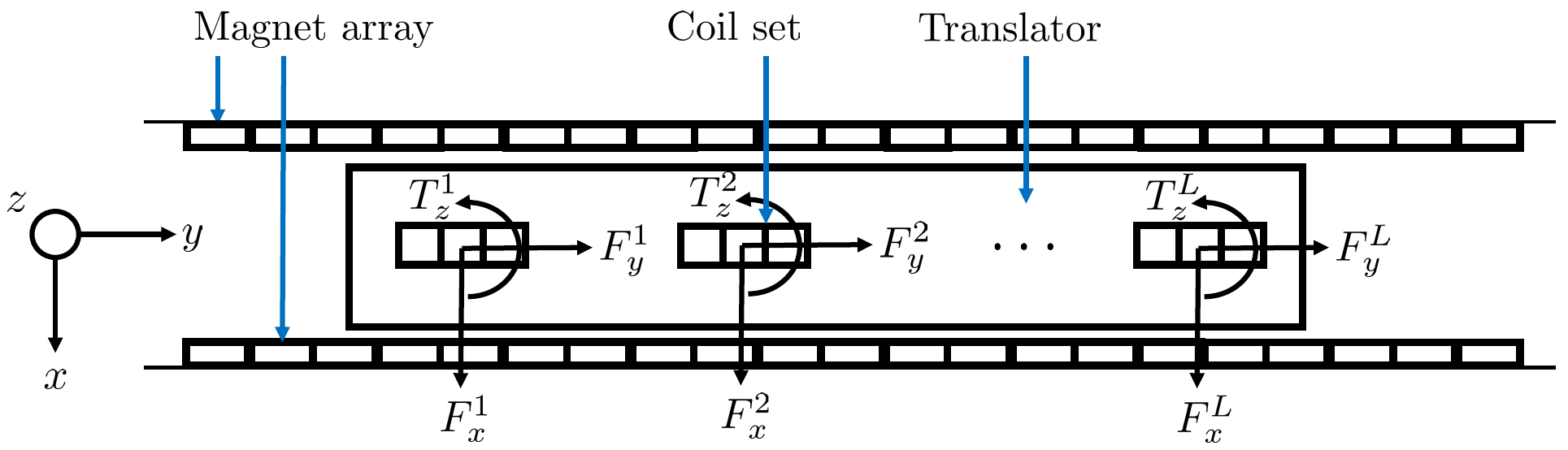}
\caption{Schematic top--view of the CLM.}
\label{fig:Schematic_Overview_CLM}
\end{figure}

The main focus in this work is the industrial CLM in Fig.~\ref{fig:CLM}, of which a schematic top--view is illustrated in Fig.~\ref{fig:Schematic_Overview_CLM}. 
However, the developed methodology can be applied to other type of actuators as well.
The CLM translates in $y$ direction and, by design, can produce forces in $y$ and $x$ direction and a torque in $z$ direction.
Forces and torques in other directions are negligible. 
Therefore, we consider the $2D$ plane in Fig.~\ref{fig:Schematic_Overview_CLM}.
As a consequence, we have $n_y = 1$ position output, and $n_F = 3$ forces $F(t) = [F_y(t), F_x(t), T_z(t)]^T$, with $F_y(t)$ the force in driving direction $y$, and $F_x(t)$ and $T_z(t)$ the out--of--plane force and torque, respectively. 
In addition, each coil set is in so--called star configuration, such that 
\begin{equation}
\label{eq:Star_Configuration}
	i_a^l(t) + i_b^l(t) + i_c^l(t) = 0, \quad l = 1, ..., L. 
\end{equation}
This constraint can be removed by reducing the number of independent inputs per coil set from $3$ to $2$, such that $i(t) \in \mathbb{R}^{2L}$. 
Note that, we have $2L$ independent inputs, which implies \emph{underactuation} of the electromagnetic part when $L=1$ and \emph{overactuation} when $L>1$. The extra $2L-3$ inputs can be used to minimize energy consumption or, e.g., not to excite flexible body modes of the translator~\cite{Custers2019}.


\subsection{Classical commutation}
We continue to derive the classical commutation solution.
\begin{assumption}
\label{ass:Lorentz_Actuation}
	The electromagnetic part follows the Lorentz princple, such that the generated force $F^l(t)$ is linear in the currents~$i^l(t)$. 
\end{assumption}
As a consequence, the electromagnetic part of a single coil set is given as
\begin{equation}
\label{eq:Electromagnetics_SingleCoilset}
	F^l(t) = K^l \big( y(t) \big) i^l(t),
\end{equation}
where $K^l : \mathbb{R}^{n_y} \rightarrow \mathbb{R}^{3 \times 3}$ is the position dependent gain matrix. 
Ideally, the system is designed to have
\begin{align}
\begin{split}
\label{eq:Electromagnetics_GainMatrix}
	&  K^l \big( y(t) \big) = \\
	& \frac{2}{3} k^l \begin{bmatrix} \sin \big( \eta^l(t) \big) & \sin \big( \eta^l(t) + \frac{2 \pi}{3} \big) & \sin \big( \eta^l(t) - \frac{2 \pi}{3} \big) \\ \mu \cos \big( \eta^l(t) \big) & \mu\cos \big( \eta^l(t) + \frac{2 \pi}{3} \big) & \mu \cos \big( \eta^l(t) - \frac{ 2 \pi}{3} \big) \\ d^l \mu \cos \big( \eta^l(t) \big) & d^l \mu \cos \big( \eta^l(t) + \frac{ 2 \pi}{3} \big) & d^l \mu \cos \big( \eta^l(t) - \frac{ 2 \pi}{3} \big)  \end{bmatrix} ,
\end{split}
\end{align}
where $k^l \in \mathbb{R}_{>0}$ is the motor constant, $d^l \in \mathbb{R}$ is the relative position of coil set $l$ with respect to the centre of mass of the translator, $\mu \in \mathbb{R}$ the ratio of the motor constant that acts in orthogonal direction, and 
\begin{equation}
\label{eq:Electromagnetics_Phase}
	\eta^l(t) = \frac{2 \pi}{d_m} y(t) + \zeta^l
\end{equation}
is the commutation phase, with $\zeta^l \in \mathbb{R}$ the commutation phase offset, and $d_m \in \mathbb{R}_{>0}$ the magnetic pole pitch, i.e., the length of two magnets. 
\begin{assumption}
\label{ass:Decoupling_Coilsets}
	The currents in coil set $l$ do not affect the force generated by coil set $m$, $l \neq m$.
\end{assumption}
Then, the total force acting on the translator is computed as the sum of the forces generated by each coil set $l$, such that
\begin{align}
\begin{split}
\label{eq:Electromagnetics_Coilsets}
	F(t) & = \sum_l F^l(t) = K\big(y(t) \big) i(t) \\
	& := [K^1 \big( y(t) \big), ..., K^L \big( y(t) \big) ] \begin{bmatrix} i^1 (t) \\ \vdots \\ i^L(t) \end{bmatrix}.
\end{split}
\end{align}

Since there are infinitely many solutions for $i(t)$ in~\eqref{eq:Electromagnetics_Coilsets} that yield $F(t) = F^*(t)$ when $L>1$, we adopt the optimal commutation formulation from~\cite{Nguyen2018}, such that
\begin{align}
\begin{split}
\label{eq:Optimal_Commutation}
	i(t) & = \textup{arg} \min_{i(t)} \| i(t) \|_2^2, \\
	& \textup{subject to: } F(t) = F^*(t).
\end{split}
\end{align}
Note that, $\| i(t) \|_2^2$ is proportional to the dissipated power. 
Substitution of~\eqref{eq:Electromagnetics_Coilsets} in~\eqref{eq:Optimal_Commutation} yields the analytical commutation solution
\begin{equation}
\label{eq:Commutation_Solution_Ideal}
	i(t) = K\big( y(t) \big)^{\dagger} F^*(t). 
\end{equation}
Suppose that $F^* (t) = [F_y^*(t), 0, 0]^T$, i.e., the out--of--plane forces are zero.
Then, the commutation~\eqref{eq:Commutation_Solution_Ideal} reduces to
\begin{align}
\begin{split}
\label{eq:Commutation_Solution_Ideal_SingleCoilset}
	i^l(t) & = \hat{\Psi}^{l^{-1}} \left( {F_y^l}^*(t) , \hat{\eta}^l(t) \right)  := \begin{bmatrix} \sin \big( \hat{\eta}^l(t) \big) \\ \sin \big( \hat{\eta}^l(t) +\frac{2 \pi}{3} \big) \\ \sin \big(\hat{\eta}^l(t) - \frac{2 \pi}{3} \big) \end{bmatrix} \frac{{F_y^l}^*(t)}{\hat{k}^l} , \\
	{F_y^l}^*(t) & = \frac{\hat{k}^{l^2}}{\sum_{i=m}^{L} ( \hat{k}^{m} )^2} F_y^*, \quad l = 1, ..., L. 
\end{split}
\end{align}
In~\eqref{eq:Commutation_Solution_Ideal_SingleCoilset}, $\hat{k}^l$ is an estimate of $k^l$, and $\hat{\eta}^l(t) = \frac{2 \pi}{d_m} y(t) + \hat{\zeta}^l$ with $\hat{\zeta}^l$ an estimate of ${\zeta}^l$. 
The desired driving force for coil set $l$ is ${F_y^l}^*(t)$, which reduces to ${F_y^l}^*(t) = \frac{1}{L} F_y^*(t)$ when $\hat{k}^l = \hat{k}$ for all $l=1, ..., L$. 

In summary, provided that the electromagnetics follows~\eqref{eq:Electromagnetics_SingleCoilset} with $K^l \big( y(t) \big)$ as in~\eqref{eq:Electromagnetics_GainMatrix}, we obtain the commutation law~\eqref{eq:Commutation_Solution_Ideal} which reduces to~\eqref{eq:Commutation_Solution_Ideal_SingleCoilset} when only a force in driving direction is desired.
Since~\eqref{eq:Commutation_Solution_Ideal_SingleCoilset} requires estimates of $\hat{k}^l$ and $\hat{\zeta}^l$, $l = 1, ..., L$, we present a novel data--based calibration technique in the next section. 
Afterwards, we demonstrate by means of real--life experiments that relying on the simplified description of the electromagnetic part~\eqref{eq:Electromagnetics_SingleCoilset},~\eqref{eq:Electromagnetics_GainMatrix} yields structural commutation errors.

\subsection{Data--based calibration}
In order to obtain accurate estimates of $\hat{k}^l$ and $\hat{\zeta}^l$ and thereby facilitate a fair comparison of classical commutation~\eqref{eq:Commutation_Solution_Ideal_SingleCoilset} with the PGNN later, we develop a data--based calibration approach. 
To do so, we generate two data sets $Z_i^l$, $i = 1,2$, for each coil set $l$ by operating only the corresponding coil set while either adding or subtracting $\Delta \in \mathbb{R}$ to the commutation phase $\hat{\zeta}^l(t)$. 
Therefore, we use ${F_y^l}^*(t) = F_y^*(t)$ and ${F_y^m}^*(t) = 0$ for $m \neq l$ in~\eqref{eq:Commutation_Solution_Ideal_SingleCoilset}, such that $F(t) = F^l(t)$ from~\eqref{eq:Electromagnetics_Coilsets}.
As a result, we obtain the data sets
\begin{align}
\begin{split}
\label{eq:DataSets}
	Z_i^l & = \{\Delta, y(1), F(1), {F_y^l}^*(1), ..., y(N), F(N), {F_y^l}^*(N) \}, \\
	&\quad \quad \quad  \quad i = 1,2, \; l = 1, ..., L.
\end{split}
\end{align}
Note that, $\Delta$ is constant throughout the experiment. 

\begin{remark}
We require force measurements to obtain the data sets as in~\eqref{eq:DataSets}. 
In the absence of force measurements, it is possible to reconstruct forces in driving direction from position measurements and an inverse model of the dynamics $G$ in Fig.~\ref{fig:Schematic_Overview}. Alternatively, we can simultaneously identify a model of the electromagnetic part $\Psi$ and a model of the dynamics $G$ in driving direction.
\end{remark}

Then, starting from an initial estimate of $\hat{k}^l$ and $\hat{\zeta}^l$, we re--calibrate the parameters following Algorithm~\ref{alg:Calibration}.
For safety reasons, the extra commutation phase offset $| \Delta | \leq \frac{\pi}{4}$. 
\begin{algorithm}
\caption{Data--based calibration of motor constant $\hat{k}^l$ and commutation phase offset $\hat{\zeta}^l$.}
\label{alg:Calibration}
	\begin{algorithmic}
		\State \textbf{Initialize} $\hat{k}^l$ and $\hat{\zeta}^l$, $l = 1, ..., L$. 
		\For{ $l = 1, ..., L$}
			\State \textbf{Generate data}~\eqref{eq:DataSets} with $\hat{\zeta}_i^l = \hat{\zeta}^l + (-1)^i \Delta$, $i=1,2$.
			\State \textbf{Optimize} coefficients: 	
			\begin{equation}
			\label{eq:Ci}
				\hat{c}_i = \textup{arg} \min_{c_i} \sum_{t \in Z_i^l} \left\| F_y(t) - c_i F_y^*(t) \right\|_2^2, \quad i = 1,2.
			\end{equation}
			\State \textbf{Update} parameters:
			\begin{align}
			\begin{split}
			\label{eq:Update_Calibration}
				\hat{\zeta}^l & \leftarrow \hat{\zeta}_1^l + \tan^{-1} \left( \frac{\frac{\hat{c}_2}{\hat{c}_1} - \cos( 2 \Delta)}{\sin(2 \Delta)} \right),  \\
				\hat{k}^l & \leftarrow \frac{\hat{c}_1 \hat{k}^l}{\cos(\hat{\zeta}^l - \hat{\zeta}_1^l)}.
			\end{split}
			\end{align}
		\EndFor 
	\end{algorithmic}
\end{algorithm}

The procedure of Algorithm~\ref{alg:Calibration} is explained by observing that the generated forces for coil set~$l$ are computed by substituting~\eqref{eq:Commutation_Solution_Ideal_SingleCoilset} in~\eqref{eq:Electromagnetics_SingleCoilset} with $K^l \big( y(t) \big)$ in~\eqref{eq:Electromagnetics_GainMatrix}, such that
\begin{equation}
\label{eq:Generated_Forces_Ideal}
	F^l(t) = \frac{k^l}{\hat{k}^l} \begin{bmatrix} \cos (\zeta - \hat{\zeta}^l) \\ \sin(\zeta^l - \hat{\zeta}^l ) \\ d^l \sin (\zeta^l - \hat{\zeta}^l ) \end{bmatrix} {F_y^l}^*(t). 
\end{equation}
As a consequence, we observe that $\hat{c}_i = \frac{k^l}{\hat{k}^l} \cos (\zeta^l - \hat{\zeta}_i^l )$ in~\eqref{eq:Ci}. 
We rewrite $\hat{c}_2$ using the trigonometric formula for the cosine as
 \begin{align}
\begin{split}
\label{eq:Derivation_Calibration}
	\hat{c}_2 &  = \frac{k^l}{\hat{k}^l} \Big( \cos( \zeta^l - \hat{\zeta}_1^l) \cos( \hat{\zeta}_1^l - \hat{\zeta}_2^l ) \\
	& \quad \quad \quad \quad -  \sin( \zeta^l - \hat{\zeta}_1^l) \sin( \hat{\zeta}_1^l - \hat{\zeta}_2^l ) \Big).
\end{split}
\end{align}
Then, from $\hat{c}_1$ and $\hat{c}_2$ we compute $k^l$ and $\zeta^l$ as in~\eqref{eq:Update_Calibration} and use these values to update the estimates $\hat{k}^l$ and $\hat{\zeta}^l$. 
Table~\ref{tab:Parameters_CLM} lists the initial and calibrated values for the CLM.

\begin{figure}
\centering
\includegraphics[width=0.95\linewidth]{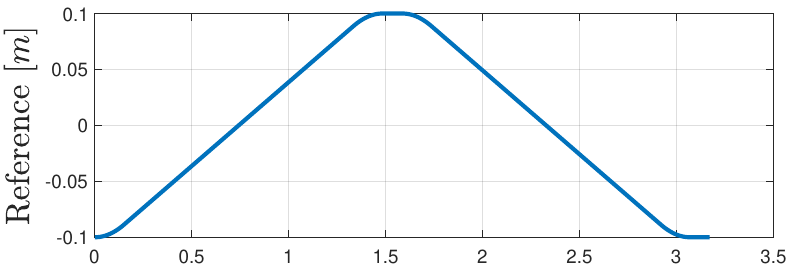}
\includegraphics[width=0.95\linewidth]{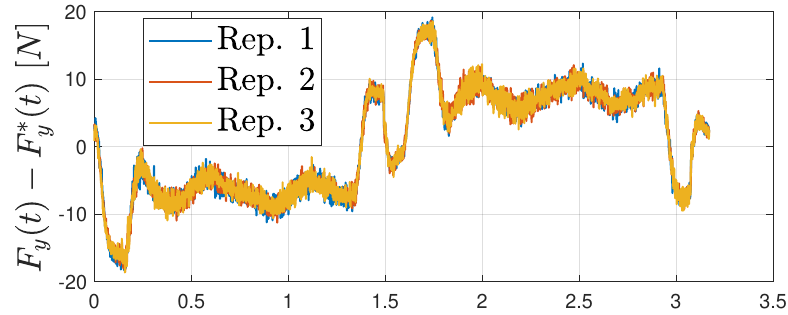}
\includegraphics[width=0.95\linewidth]{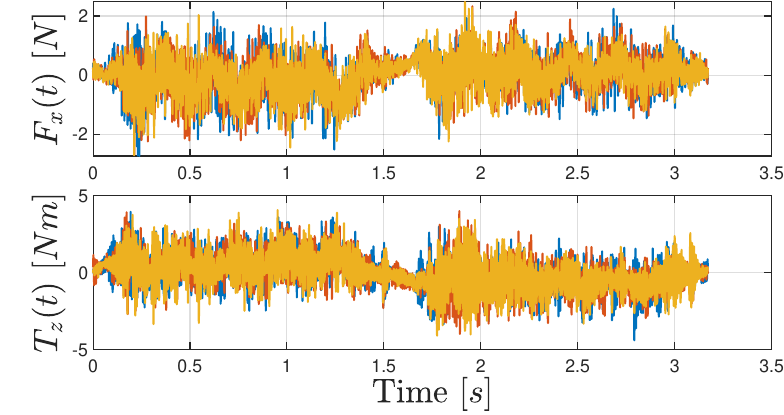}
\caption{Position reference $y^*(t)$ (top window), and the commutation error in driving direction $F_y(t) - F_y^*(t)$ measured on the CLM using commutation law~\eqref{eq:Commutation_Solution_Ideal} with data--based calibrated parameters.}
\label{fig:Commutation_Errors_Problem}
\end{figure}

\begin{table}
\centering
\caption{Initial and calibrated values of $\hat{k}^l$ and $\hat{\zeta}^l$.}
\label{tab:Parameters_CLM}
\begin{tabular}{c || c | c}
	\textbf{Parameter} & \textbf{Initial} & \textbf{Calibrated} \\ \hline \hline 
	$\hat{k}^1$ & & $61.34$ \\
	$\hat{k}^2$ & $67$ & $61.62$ \\ 
	$\hat{k}^3$ & & $60.07$ \\ \hline 
	$\hat{\zeta}^1$ & & $-0.54$ \\ 
	$\hat{\zeta}^2$ & $-0.52$ & $-0.55$ \\ 
	$\hat{\zeta}^3$ & & $-0.55$ \\ \hline
\end{tabular}
\end{table}

\section{Problem statement}
\label{sec:Problem_Statement}
Next, we demonstrate the inherent limitations of the commutation algorithm~\eqref{eq:Commutation_Solution_Ideal},~\eqref{eq:Commutation_Solution_Ideal_SingleCoilset} that is derived from the assumption that the electromagnetic part satisfies~\eqref{eq:Electromagnetics} with $K^l\big(y(t)\big)$ in~\eqref{eq:Electromagnetics_GainMatrix}. 
To facilitate a fair comparison, we use Algorithm~\ref{alg:Calibration} with $\Delta = \frac{\pi}{4}$ to calibrate $\hat{k}^l$ and $\hat{\zeta}^l$. 
The motor constants $\hat{k}^l $ and commutation phase offset $\hat{\zeta}^l $ are initialized from manufacturer specifications, and the initialization in~\cite{Rovers2002}, see Table~\ref{tab:Parameters_CLM} for the initial and calibrated values.

Fig.~\ref{fig:Commutation_Errors_Problem} shows the commutation error $F(t) - F^*(t)$ on the experimental CLM resulting from the standard commutation law~\eqref{eq:Commutation_Solution_Ideal_SingleCoilset} with calibrated parameters $\hat{k}^l$ and $\hat{\zeta}^l$. 
The desired force is $F^*(t) = [F_y^*(t), 0, 0]^T$, where $F_y^*(t)$ results from the position control loop.
We observe that, especially in driving direction, there is a significant commutation error which is highly reproducible.
This commutation error results from unknown and hard--to--model effects present in the electromagnetic part~\eqref{eq:Electromagnetics} that are not described by the ideal model~\eqref{eq:Electromagnetics_SingleCoilset},~\eqref{eq:Electromagnetics_GainMatrix} for which the commutation law~\eqref{eq:Commutation_Solution_Ideal_SingleCoilset} is the analytical solution. 

Therefore, in this work, we develop a PGNN--based model for identification of the electromagnetic part~\eqref{eq:Electromagnetics} with higher accuracy than the classical commutation approach.
This PGNN model is used afterwards to derive an analytical commutation law, which achieves higher performance compared to classical commutation~\eqref{eq:Commutation_Solution_Ideal_SingleCoilset}.
In addition, motivated by common limitation in industrial applications where the commutation law~$\hat{\Psi}^{l^{-1}}$ in~\eqref{eq:Commutation_Solution_Ideal_SingleCoilset} is fixed, we design an input transformation to control the commutation magnitude ${F_y^l}^*(t)$ and phase $\hat{\eta}^l(t)$ instead of the currents $i^l(t)$.

In the remainder of this paper, we fix $\hat{k}^l = \hat{k}$ and $\hat{\zeta}^l = \hat{\zeta}$, $l = 1, ..., L$, in the commutation $\hat{\Psi}^{l^{-1}}$ in~\eqref{eq:Commutation_Solution_Ideal_SingleCoilset}.
As it will be demonstrated, the developed PGNN commutation will inherently compensate for mismatches in the motor constant and commutation phase offset. 
Additionally, due to the star configuration~\eqref{eq:Star_Configuration}, we consider $i^l(t) = [i_a^l(t), i_b^l(t)]^T \in \mathbb{R}^2$.
The current $i_c^l(t)$, when necessary, is computed as $i_c^l(t) = - i_a^l(t) - i_b^l(t)$.

\section{Physics--guided neural networks for commutation}
\label{sec:Commutation}
In the following, we consider a new set of inputs ${F_y^l}^*(t)$ and~$\Delta^l(t)$ when $i^l(t)$ cannot be controlled directly, i.e., when $\hat{\Psi}^{l^{-1}}$ in~\eqref{eq:Commutation_Solution_Ideal_SingleCoilset} is fixed.
Then, we develop a PGNN--based model to describe the electromagnetic part, which is used afterwards to derive an analytical commutation law.

\begin{figure}
\centering
\includegraphics[width=0.8\linewidth]{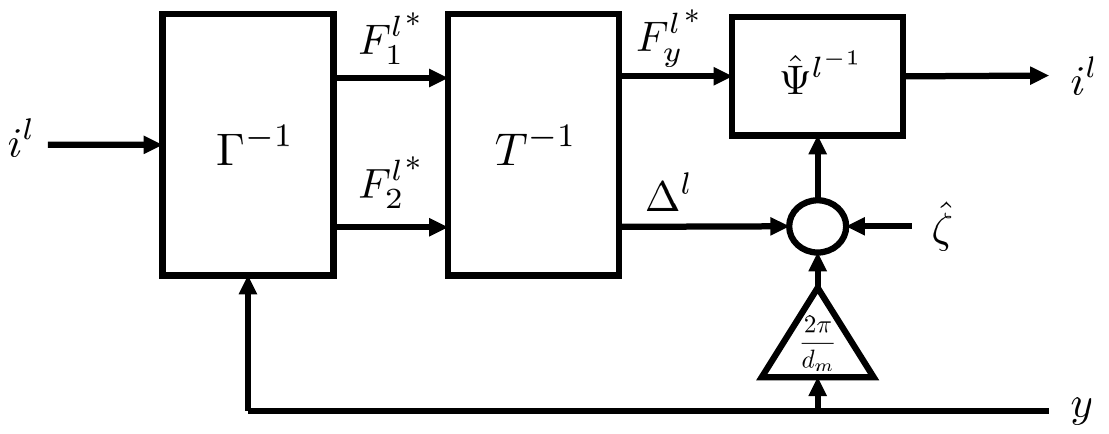}
\caption{Schematic overview to control $i^l(t)$ via ${F_y^l}^*(t)$ and $\Delta^l(t)$.}
\label{fig:Actuator_Inputs}
\end{figure}

\subsection{Suitable selection of control inputs}
\label{subsec:Inputs}
As previously stated, we resort to actively controlling the commutation magnitude ${F_y^l}^*(t)$ and phase~$\hat{\eta}^l(t)$ instead of two independent currents $i_a^l(t)$ and $i_b^l(t)$.
To do so, we adopt a change of inputs from $\{ i_a^l(t) , i_b^l(t) \}$ to $\{ {F_y^l}^*(t), \Delta^l(t) \}$, where $\Delta^l(t)$ is an offset to the commutation phase as visualized in Fig.~\ref{fig:Actuator_Inputs}. 
As a result, the commutation~\eqref{eq:Commutation_Solution_Ideal_SingleCoilset} becomes
\begin{align}
\begin{split}
\label{eq:Commutation_Solution_With_PhaseOffset}
	i^l(t) & = \hat{\Psi}^{l^{-1}} \left( {F_y^l}^*(t), \hat{\eta}(t) + \Delta^l(t) \right). \\
\end{split}
\end{align}
We rewrite the commutation law~\eqref{eq:Commutation_Solution_With_PhaseOffset} with $\hat{\Psi}^{l^{-1}}$ in~\eqref{eq:Commutation_Solution_Ideal_SingleCoilset} using the sum--formula for the sine term as follows
\begin{align}
\begin{split}
\label{eq:Commutation_Rewritten1}
	i^l(t) & = \frac{1}{\hat{k}} \begin{bmatrix} \sin \big( \hat{\eta} (t) + \Delta^l(t) \big) \\ \sin \big( \hat{\eta}(t) + \Delta^l(t) + \frac{2 \pi}{3} \big)  \end{bmatrix} {F_y^l}^*(t) = \Gamma \big( y(t) \big) \begin{bmatrix} {F_1^l}^*(t) \\ {F_2^l}^*(t) \end{bmatrix}. 
\end{split}
\end{align}
In~\eqref{eq:Commutation_Rewritten1}, $\Gamma \big( y(t) \big)$ is given as
\begin{equation}
\label{eq:Commutation_Rewritten1_Gamma}
	\Gamma \big( y(t) \big) := \frac{1}{\hat{k}} \begin{bmatrix} \sin  \big( \hat{\eta}(t) \big) & \cos \big( \hat{\eta}(t) \big)  \\ \sin  \big( \hat{\eta}(t) + \frac{2 \pi}{3} \big) & \cos \big( \hat{\eta}(t) + \frac{2 \pi}{3} \big)  \end{bmatrix} ,
\end{equation}
and the input transformation, denoted by $T$, and its inverse are
\begin{align}
\begin{split}
\label{eq:Input_Transformation}
	T: \quad & \begin{bmatrix} {F_1^l}^*(t) \\ {F_2^l}^*(t) \end{bmatrix} = \begin{bmatrix} \cos \big( \Delta^l(t) \big) {F_y^l}^*(t) \\ \sin \big( \Delta^l (t)\big) {F_y^l}^*(t)  \end{bmatrix} , \\
	T^{-1} : \quad & \begin{bmatrix} {F_y^l}^*(t) \\ \Delta^l(t)  \end{bmatrix} = \begin{bmatrix} \sqrt{ \big( {{F_1^l}^*}(t) \big)^2 + \big( {{F_2^l}^*}(t) \big)^2} \\ \tan^{-1} \left( \frac{{F_2^l}^*(t) }{{F_1^l}^*(t) } \right)  \end{bmatrix}. 
\end{split}
\end{align}

In summary, the commutation yields a set of currents $i^l(t)$ to be supplied to the system. 
Often, we are restricted to use the commutation magnitude and phase $\{ {F_y^l}^*(t), \hat{\eta}(t)+\Delta^l(t) \}$ to the system.
The magnitude and the phase $\{ {F_y^l}^*(t), \hat{\eta}(t)+\Delta^l(t) \}$ that yield the current $i^l(t)$ can be computed via the transformations~\eqref{eq:Commutation_Rewritten1} and~\eqref{eq:Input_Transformation}, i.e., we premultiply $i^l(t)$ by $\Gamma \big( y(t) \big)^{-1}$ to obtain $\{ {F_1^l}^*(t), {F_2^l}^*(t) \}$ followed by $T^{-1}$ to obtain $\{ {F_y^l}^*(t), \Delta^l(t) \}$. 
The procedure is schematically visualized in Fig.~\ref{fig:Actuator_Inputs}.

\subsection{PGNNs for identification of the electromagnetic part}
\label{subsec:PGNN_Electromagnetics}
\begin{figure}
\centering
\includegraphics[width=1.0\linewidth]{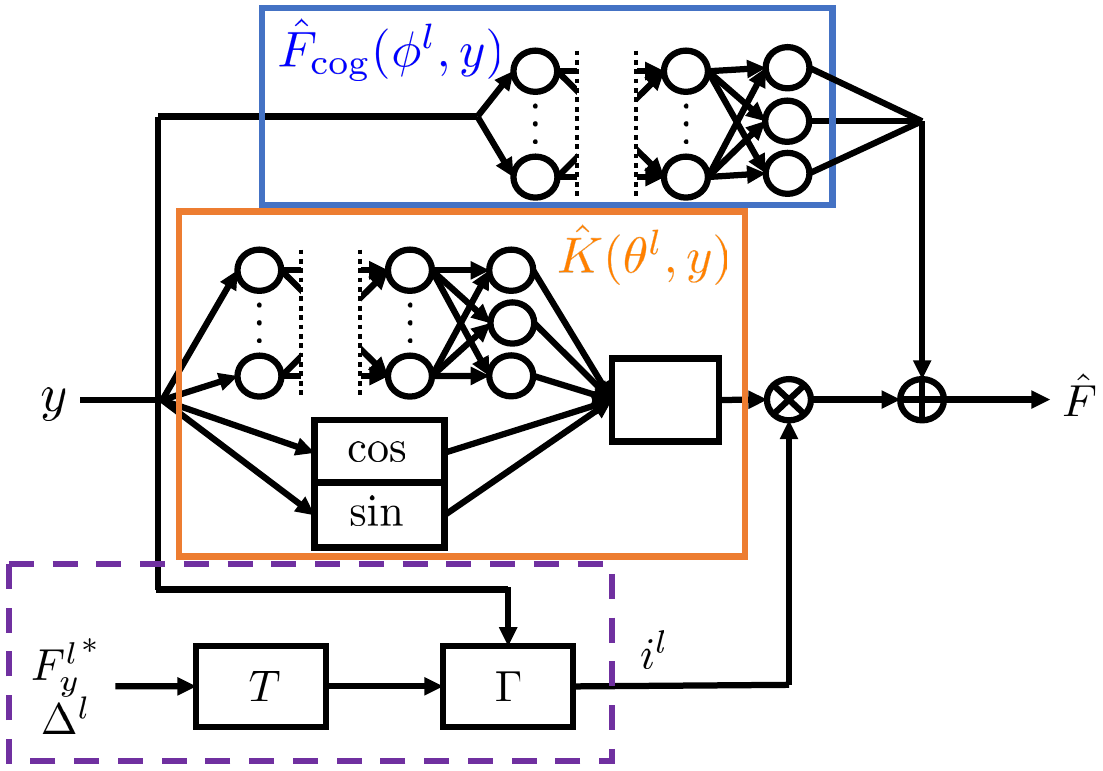}
\caption{Schematic overview of the PGNN for identification of the electromagnetic part. When currents $i^l$ are directly accessible and present in the data sets~$Z_{1,2}^l$, the dotted purple block is removed.}
\label{fig:PGNN_Parametrization}
\end{figure}
We adopt Assumptions~\ref{ass:Lorentz_Actuation} and~\ref{ass:Decoupling_Coilsets} on the electromagnetic part~\eqref{eq:Electromagnetics}, such that we obtain
\begin{align}
\begin{split}
\label{eq:Electromagnetics_WithAssumptions}
	F(t) &= \Psi \big( i(t), y(t) \big) = \sum_l K^l \big( y(t) \big) i^l(t) + F_{\textup{cog}} \big( y(t) \big),
\end{split}
\end{align}
where $K^l \big( y(t) \big) : \mathbb{R} \rightarrow \mathbb{R}^{3 \times 2}$ is the gain matrix, and $F_{\textup{cog}}: \mathbb{R} \rightarrow \mathbb{R}^3$ the cogging force.
Both $K^l \big( y(t) \big)$ and $F_{\textup{cog}} \big( y(t) \big)$ are (partly) unknown. 
Therefore, we develop a PGNN--based model of the electromagnetic part~\eqref{eq:Electromagnetics_WithAssumptions} which enhances the physics--based knowledge structurally with neural networks.
Consider again the data sets $Z_{1,2}^l$ in~\eqref{eq:DataSets}, i.e., generated by operating a single coil set. 
Then, we predict the measured force $F(t)$ with the PGNN as
\begin{align}
\begin{split}
\label{eq:PGNN_Electromagnetics_SingleCoilset}
	\hat{F}&(\theta^l, \phi^l, t) = \hat{K} \big( \theta^l, y(t) \big) i^l(t) + \hat{F}_{\textup{cog}} \big( \phi^l , y(t) \big) \\
	& = \hat{K} \big( \theta^l , y(t) \big) \Gamma \big( y(t) \big) T \left( \begin{bmatrix} {F_y^l}^*(t) \\ \Delta^l(t) \end{bmatrix} \right) + \hat{F}_{\textup{cog}} \big( \phi^l , y(t) \big),
\end{split}
\end{align}
where $\theta^l \in \mathbb{R}^{n_{\theta}}$ and $\phi^l \in \mathbb{R}^{n_{\phi}}$ are the parameters corresponding to the gain matrix and cogging force, $n_{\theta}, n_{\phi} \in \mathbb{Z}_{>0}$. 
The PGNN model~\eqref{eq:PGNN_Electromagnetics_SingleCoilset} is visualized in Fig.~\ref{fig:PGNN_Parametrization}, and the gain matrix and cogging force are given as
\begin{align}
\begin{split}
\label{eq:GainMatrix_Cogging}
	\hat{K} \big( \theta^l, y(t) \big) & = A^l \cos \left( \frac{ 2 \pi}{d_m} y(t) \right) + B^l \sin \left( \frac{2 \pi}{d_m} y(t) \right) \\
	& \quad \quad + \begin{bmatrix} f_{\textup{NN}} \big( \theta_1^l, y(t)  \big) & f_{\textup{NN}} \big( \theta_2^l, y(t) \big) \end{bmatrix} \\
	\hat{F}_{\textup{cog}} \big( \phi^l, y(t) \big) & = f_{\textup{NN}} \big( \phi^l, y(t) \big) .
\end{split}
\end{align}
In~\eqref{eq:GainMatrix_Cogging}, $A^l, B^l \in \mathbb{R}^{3 \times 2}$ are the physical parameters, such that $\theta_{\textup{phy}}^l = [\textup{col}(A^l)^T, \textup{col}(B^l)^T]^T$ and $\theta^l = [{\theta_{\textup{phy}}^l}^T, {\theta_1^l}^T, {\theta_2^l}^T]^T$.
The parameter vectors $\theta_1^l$, $\theta_2^l$ and $\phi^l$ represent the weights and biases of the corresponding neural network.
Let $\mu \in \{ \theta_1^l, \theta_2^l, \phi^l \}$, then the network output is given as
\begin{align}
\begin{split}
\label{eq:NeuralNetwork}
	f_{\textup{NN}} \big( \mu, y(t) \big) = W_{I+1}^{\mu} \alpha_I \bigg( ... \alpha_1 \big( W_1^{\mu} y(t) + B_1^{\mu} \big) \bigg) + B_{I+1}^{\mu},
\end{split}
\end{align}
where $W_{i}^{\mu} \in \mathbb{R}^{n_{i-1} \times n_i}$ are the weights, $B_i^{\mu} \in \mathbb{R}^{n_i}$ the biases, $\alpha_i : \mathbb{R}^{n_i} \rightarrow \mathbb{R}^{n_i}$ denotes the aggregation of activation functions, $n_i \in \mathbb{Z}_{>0}$ the number of neurons in layer $i = 1, ..., I$ with $I \in \mathbb{Z}_{>0}$ the number of hidden layers. The parameters are given as $\mu = [\textup{col}(W_1^{\mu})^T, {B_1^{\mu}}^T, ..., \textup{col}(W_I^{\mu})^T, {B_I^{\mu}}^T]^T$.

\begin{remark}
	To ensure a persistently exciting data set for the PGNN~\eqref{eq:PGNN_Electromagnetics_SingleCoilset},~\eqref{eq:GainMatrix_Cogging}, we require for all positions $y(t)$ the following:
	\begin{enumerate}
		\item Different magnitude of the desired force ${F_y^l}^*(t)$ to separate the Lorentz force from the cogging force;
		\item Different ratio of the currents $i_a^l(t)$ and $i_b^l(t)$ to separate the columns of $\hat{K}^l \big( \theta^l, y(t) \big)$.  
	\end{enumerate}
	Therefore, we generate $Z_{i}^l$, $i=1,2$, with a reference $y^*(t)$ that encompasses the full stroke of the system with different velocities (point $1.$, assuming non--zero friction in the system) and repeat the experiment with $\Delta^l(t) = (-1)^i \Delta$, $i=1,2$ (point $2.$). 
\end{remark}
\begin{remark}
	From the physical understanding of the system, we expect position dependent harmonic functions with a period of ${d_m}$, as in~\eqref{eq:Electromagnetics_GainMatrix}.
	This physics--based information is embedded structurally in the PGNN model~\eqref{eq:PGNN_Electromagnetics_SingleCoilset}.
	The addition of the neural networks allow the PGNN~\eqref{eq:PGNN_Electromagnetics_SingleCoilset},~\eqref{eq:GainMatrix_Cogging} to include complex to model and unknown effects, for example caused by an imperfect magnetic field due to differences in the field intensity or orientation of the permanent magnets.
\end{remark}

The parameters of the PGNN are identified for each coil set using the data sets $Z_1^l$ and $Z_2^l$ as in~\eqref{eq:DataSets} according to Algorithm~\ref{alg:Identification}.
Note that,~$\Delta$ can be considered as an excitation for $\Delta^l(t)$ in~\eqref{eq:Commutation_Solution_With_PhaseOffset}. 
The cost function is chosen as the sum of the MSE of the data fit and a regularizatoin term, such that
\begin{align}
\begin{split}
\label{eq:Identification_CostFunction}
	V \left( \begin{bmatrix} \theta^l \\ \phi^l \end{bmatrix}, Z_1^l, Z_2^l \right) & = \frac{1}{2N} \sum_{t \in Z_{1,2}^l } \| F(t) - \hat{F} (\theta^l, \phi^l, t) \|_2^2 \\
	& \quad \quad \quad  + \left\| \Lambda \left(\theta_{\textup{phy}}^l - \theta_{\textup{phy}}^{l^*} \right) \right\|_2^2.
\end{split}
\end{align}
The regularization term in~\eqref{eq:Identification_CostFunction} is proposed in~\cite{Bolderman2024} and solves the overparametrization in the PGNN that is caused by the ability of the NN to identify also the physics--based part of~$\hat{K}^l \big( \theta^l, y(t) \big)$ in~\eqref{eq:GainMatrix_Cogging}.
The variable $\theta_{\textup{phy}}^{l^*}$ describes a set of desired physical parameters, e.g., obtained by identification of the stand--alone physical model or from the ideal description of the electromagnetic part~\eqref{eq:Electromagnetics_GainMatrix}. The matrix $\Lambda \in \mathbb{R}^{12 \times 12}$ quantifies the weight of the regularization.

\begin{algorithm}
\caption{Identification of the electromagnetic part~\eqref{eq:Electromagnetics_WithAssumptions} with PGNN--based model~\eqref{eq:PGNN_Electromagnetics_SingleCoilset}.}
\label{alg:Identification}
	\begin{algorithmic}
		\State \textbf{Initialize} $\hat{k}$ and $\hat{\zeta}$. 
		\For{ $l = 1, ..., L$}
			\State \textbf{Generate data}~\eqref{eq:DataSets} using commutation law~\eqref{eq:Commutation_Solution_With_PhaseOffset} with
			\State \quad \quad  $\Delta^l(t) = (-1)^i \Delta$, $\, i = 1,2$. 
			\State \textbf{Identify} parameters: 	
			\begin{equation}
			\label{eq:Identification_Criterion}
				\begin{bmatrix} \hat{\theta}^l \\ \hat{\phi}^l \end{bmatrix} = \textup{arg} \min_{[{\theta^l}^T, {\phi^l}^T]^T} V \left( \begin{bmatrix} \theta^l \\ \phi^l \end{bmatrix} , Z_1^l, Z_2^l \right). 
			\end{equation}
		\EndFor 
	\end{algorithmic}
\end{algorithm}

\textcolor{black}{Training the PGNN~\eqref{eq:PGNN_Electromagnetics_SingleCoilset} according to identification criterion~\eqref{eq:Identification_Criterion} with cost function~\eqref{eq:Identification_CostFunction} is a non--convex optimization, and potentially yields a local minimum. 
With the aim to ensure that the PGNN~\eqref{eq:PGNN_Electromagnetics_SingleCoilset} achieves a smaller cost function~\eqref{eq:Identification_CostFunction} compared to the classical model~\eqref{eq:Electromagnetics_Coilsets},~\eqref{eq:Electromagnetics_GainMatrix}, we derive a least--squares solution of~\eqref{eq:Identification_CostFunction} with respect to the parameters in which the PGNN is linear.
The least--squares solution can be used before training as an initialization, during training, or after training.
To do so, we rewrite the PGNN~\eqref{eq:PGNN_Electromagnetics_SingleCoilset},~\eqref{eq:GainMatrix_Cogging} as
\begin{align}
\begin{split}
\label{eq:PGNN_Rewritten_LIP}
	\hat{F} \big( \theta^l, \phi^l, t) = \begin{bmatrix} {\theta_{\textup{L}, y}^l}^T M_y (\theta_{\textup{NL}}^l, t) \\ {\theta_{\textup{L}, x}^l}^T M_x (\theta_{\textup{NL}}^l, t) \\ {\theta_{\textup{L}, z}^l}^T M_z (\theta_{\textup{NL}}^l, t) \end{bmatrix}. 
\end{split}
\end{align}
In~\eqref{eq:PGNN_Rewritten_LIP}, $\theta_{\textup{L}, q}^l$, $q =  y, x, z$, are the parameters in which the PGNN is linear, and $M_q (\theta_{\textup{NL}}^l, t)$ form the corresponding regressor vectors. 
More specifically, $\theta_{\textup{L},q}^l$ is given as
\begin{align}
\begin{split}
\label{eq:Linear_Parameter_Vector}
	\theta_{\textup{L}}^l & = [A^l, B^l, W_{I+1}^{\theta_1^l}, B_{I+1}^{\theta_1^l}, W_{I+1}^{\theta_2^l}, B_{I+1}^{\theta_2^l}, W_{I+1}^{\phi^l}, B_{I+1}^{\phi^l}]^T, \\
	\theta_{\textup{L},y}^l &= \theta_{\textup{L}}^l \begin{bmatrix} 1 \\ 0 \\ 0 \end{bmatrix}, \quad \theta_{\textup{L},x}^l = \theta_{\textup{L}}^l \begin{bmatrix} 0 \\ 1 \\ 0 \end{bmatrix}, \quad \theta_{\textup{L},z}^l = \theta_{\textup{L}}^l \begin{bmatrix} 0 \\ 0 \\ 1  \end{bmatrix}.
\end{split}
\end{align}
Suppose that $\Lambda$ is diagonal. Then, the cost function~\eqref{eq:Identification_CostFunction} can be rewritten as
\begin{align}
\begin{split}
\label{eq:Identification_CostFunctionRewritten}
	V & \left( \begin{bmatrix} \theta^l \\ \phi^l \end{bmatrix}, Z_1^l, Z_2^l \right) = \sum_{q \in \{y,x,z\} } \Bigg( \frac{1}{2N} \sum_{t \in Z_{1,2}^l} \Big( F_q(t) - {\theta_{\textup{L},q}^l}^T M_q (\theta_{\textup{NL}}^l, t) \Big)^2 \\
	& \quad \quad \quad \quad \quad   + \left\| \Lambda_q [I, 0] \left( \theta_{\textup{L},q}^l - \begin{bmatrix} \theta_{\textup{phy},q}^{l^*} \\ 0 \end{bmatrix} \right)\right\|_2^2  \Bigg) .
\end{split}
\end{align}
As a consequence, the parameters $\theta_{\textup{L},q}^l$ can be chosen as the least--squares solution for each direction independently. 
By doing so, the PGNN achieves a lower value of the cost function~\eqref{eq:Identification_CostFunction} compared to the classical model, as is formalized in the next proposition. 
\begin{proposition}
\label{prop:Initialization}
Consider the PGNN~\eqref{eq:PGNN_Electromagnetics_SingleCoilset},~\eqref{eq:GainMatrix_Cogging} written as~\eqref{eq:PGNN_Rewritten_LIP} with identification cost function~\eqref{eq:Identification_CostFunction}. 
Let~$\Lambda$ be diagonal and choose $\theta_{\textup{phy}}^{l^*}$ to represent the classical model~\eqref{eq:Electromagnetics_SingleCoilset},~\eqref{eq:Electromagnetics_GainMatrix}. 
Suppose that $\theta_{\textup{NL}}^l$ is known, e.g., randomly initialized or attained during training, and choose 
\begin{align}
\begin{split}
\label{eq:Optimal_Initialization}
	\theta_{\textup{L},q}^l & = \left( \frac{1}{2N} \sum_{t \in Z_{1,2}^l} M_q(\theta_{\textup{NL}}^l, t) M_q (\theta_{\textup{NL}}^l, t)^T + \begin{bmatrix} \Lambda_q^2 & 0 \\ 0 & 0 \end{bmatrix} \right)^{-1} \cdot  \\
	& \left( \frac{1}{2N} \sum_{t \in Z_{1,2}^l} F_q(t) M_q (\theta_{\textup{NL}}^l, t) + \begin{bmatrix} \Lambda_q^2 & 0 \\ 0 & 0 \end{bmatrix} \begin{bmatrix} \theta_{\textup{phy},q}^{l^*} \\ 0 \end{bmatrix}  \right) , \; \; q  = y, x, z . 
\end{split}
\end{align}
Then, it holds that
\begin{equation}
\label{eq:LowerCostFunction}
	V \left( \begin{bmatrix} \theta^l \\ \phi^l \end{bmatrix} , Z_1^l, Z_2^l \right) \leq V \left( \theta_{\textup{phy}}^{l^*}, Z_1^l, Z_2^l \right).
\end{equation}
Moreover,~\eqref{eq:LowerCostFunction} holds with strict inequality if and only if
\begin{align}
\begin{split}
\label{eq:StrictInequality_Condition}
	& \frac{1}{2N} \sum_{t \in Z_{1,2}^l} M_q(\theta_{\textup{NL}}^l, t) \left( F_q(t) - M_q(\theta_{\textup{NL}}^l, t)^T \begin{bmatrix} \theta_{\textup{phy},q}^{l^*} \\ 0 \end{bmatrix} \right) \neq 0 , \\
	& \quad \quad \quad  \textup{for some } q \in \{y,x,z\}. 
\end{split}
\end{align}
\end{proposition}
\begin{proof}
	The PGNN~\eqref{eq:PGNN_Rewritten_LIP} recovers the classical model~\eqref{eq:Electromagnetics_Coilsets},~\eqref{eq:Electromagnetics_GainMatrix} for
	\begin{align}
	\begin{split}
	\label{eq:Recover_ClassicalModel}
		\theta_{\textup{L}}^{l^*} = [A^{l^*}, B^{l^*}, 0]^T.
	\end{split}
	\end{align}
	Eq.~\eqref{eq:Optimal_Initialization} is the least--squares solution of the cost function~\eqref{eq:Identification_CostFunctionRewritten} in $q$ direction, i.e., the global minimum of the cost function~\eqref{eq:Identification_CostFunctionRewritten} for the parameters $\theta_{\textup{L},q}^l$. 
	Therefore, the classical model, recovered for $\theta_{\textup{L},q}^l$ as in~\eqref{eq:Recover_ClassicalModel}, gives the upperbound~\eqref{eq:LowerCostFunction}. 
	The upperbound~\eqref{eq:LowerCostFunction} holds with equality only if $\theta_{L,q}^l - \begin{bmatrix} \theta_{\textup{phy},q}^{l^*} \\ 0 \end{bmatrix} = 0$ for all $q \in y,x,z$, which yields condition~\eqref{eq:StrictInequality_Condition}.  
\end{proof}
}

The PGNN model of the complete electromagnetic part of the system is obtained by substituting $[{\theta^l}^T, {\phi^l}^T]^T = [\hat{\theta}^{l^T}, \hat{\phi}^{l^T}]^T$ in~\eqref{eq:PGNN_Electromagnetics_SingleCoilset} and combining the coil sets as follows
\begin{align}
\begin{split}
\label{eq:PGNN_Electromagnetics}
	\hat{F}(\hat{\theta}, \hat{\phi}, t) & = \sum_{l} \hat{K}^l \big( \hat{\theta}^l , y(t) \big) i^l(t) + \frac{1}{L} \sum_{l} \hat{F}_{\textup{cog}} \big( \hat{\phi}^l, y(t) \big) \\
	& = \hat{K} \big( \hat{\theta}, y(t) \big) \begin{bmatrix} i^1 (t) \\ \vdots \\ i^L(t) \end{bmatrix} + \frac{1}{L} \sum_{l} \hat{F}_{\textup{cog}} \big( \hat{\phi}^l, y(t) \big). 
\end{split}
\end{align}
Note that, in~\eqref{eq:PGNN_Electromagnetics} we take the mean of the cogging forces, since the total cogging force is identified when performing the identification for each coil set. 

\begin{remark}
	In the absence of force measurements, it is possible to reconstruct the forces from position measurements $y$ and an inverse model of the dynamics $G$ in Fig.~\ref{fig:Schematic_Overview}. 
	Alternatively, we can simultaneously identify a model of the electromagnetic part $\Psi$ and a model of the dynamics $G$. 
	However, these approaches are limited to directions in which movement is measured.
\end{remark}

\subsection{An analytical commutation solution}
\label{subsec:PGNN_Commutation}
A PGNN model~\eqref{eq:PGNN_Electromagnetics} describing the electromagnetic part~\eqref{eq:Electromagnetics} is identified in the previous section. 
In the following, we formulate the optimal commutation problem~\eqref{eq:Optimal_Commutation} using the PGNN model~\eqref{eq:PGNN_Electromagnetics} based on which we derive an analytical solution which is suitable for real--time implementation. 
We replace the constraints in~\eqref{eq:Optimal_Commutation} to have the predicted forces $\hat{F}$ become the desired forces $F^*$, such that we obtain 
\begin{align}
\begin{split}
\label{eq:Optimal_Commutation_Model_Current}
	i(t) & = \textup{arg} \min_{i(t)} \| i(t) \|_2^2 , \\
	& \textup{subject to: } \hat{F}(\hat{\theta}, \hat{\phi},t) = F^*(t). 
\end{split}
\end{align}
Suppose that we are able to control directly the currents. Then, the commutation law is obtained by substituting the electromagnetic part~\eqref{eq:PGNN_Electromagnetics} in~\eqref{eq:Optimal_Commutation_Model_Current}, such that
\begin{equation}
\label{eq:Commutation_Model_Currents}
	i(t) =  \hat{K} \big( \hat{\theta}, y(t) \big) ^{\dagger} \Big( F^*(t) - \frac{1}{L} \sum_l \hat{F}_{\textup{cog}} \big( \hat{\phi}^l, y(t) \big) \Big).
\end{equation}
\textcolor{black}{
\begin{remark} 
The optimal commutation problem~\eqref{eq:Optimal_Commutation_Model_Current} is analytically solvable due to the adopted PGNN structure~\eqref{eq:PGNN_Electromagnetics}, i.e., by learning position dependent coefficients with the NNs.
This is novel compared to the PGNN architectures used for feedforward control~\cite{Bolderman2024}, where the NNs are additive to the output.
Such a PGNN for the electromagnetic part is given as
\begin{align}
\begin{split}
\label{eq:Electromagnetics_PGNN_Example}
	\hat{F}^l (\theta_{\textup{phy}}^l, \theta_{\textup{NN}}^l, t) =&  \left( A^l \cos \left( \frac{ 2 \pi}{d_m} y(t) \right) + B^l \sin \left( \frac{ 2 \pi}{d_m} y(t) \right) \right) i^l(t) \\
	&  + f_{\textup{NN}} \big( \theta_{\textup{NN}}^l, y(t), i^l(t) \big),
\end{split}
\end{align}
for which no analytical solution can be computed from~\eqref{eq:Optimal_Commutation_Model_Current}, which limits practical applicability. 
\end{remark}}

As a final step, when necessary, we can transform the currents $i^l(t)$ to the commutation magnitude ${F_y^l}^*(t)$ and phase $\Delta^l(t)$.
This is achieved by adopting the scheme in Fig.~\ref{fig:Actuator_Inputs}, i.e., we use $\Gamma \big( y(t) \big)$ in~\eqref{eq:Commutation_Rewritten1_Gamma} and $T^{-1}$ in~\eqref{eq:Input_Transformation} such that
\begin{equation}
\label{eq:Current_To_ForceDelta}
	\begin{bmatrix} {F_y^l}^*(t) \\ \Delta^l(t) \end{bmatrix} = T^{-1} \left( \Gamma \big( y(t) \big)^{-1} i^l(t) \right).
\end{equation}
Observe that ${F_y^l}^*(t)$ and $\Delta^l(t)$ are the commutation inputs that yield $i^l(t)$, see Fig.~\ref{fig:Actuator_Inputs}. 
The resulting closed--loop control setup with the proposed PGNN commutation algorithm is visualized in Fig.~\ref{fig:Schematic_Overview_Proposed}.

\begin{figure}
\centering
\includegraphics[width=1.0\linewidth]{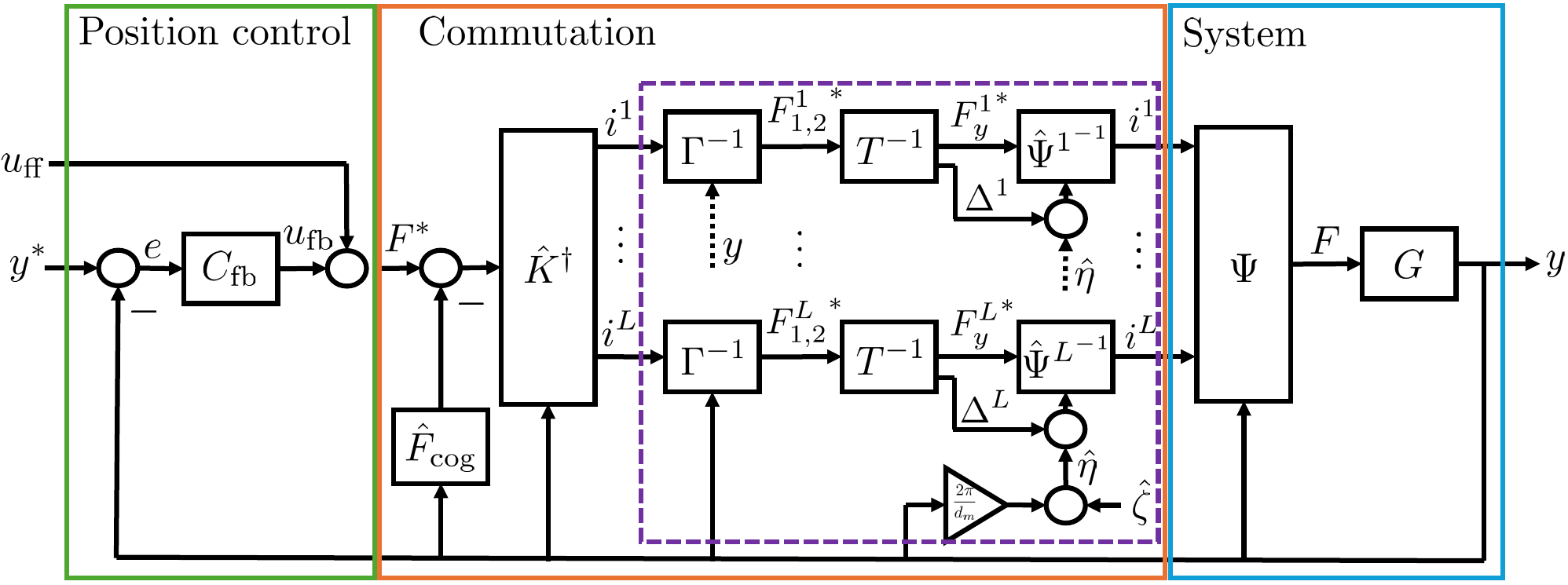}
\caption{Schematic overview of the proposed PGNN--based commutation algorithm. When currents $i^l$ are directly prescribed, the dotted purple block can be removed.}
\label{fig:Schematic_Overview_Proposed}
\end{figure}

\section{Experimental results}
\label{sec:Experiments}
We demonstrate the effectiveness of the developed PGNN commutation framework on the real--life industrial CLM shown in Fig.~\ref{fig:CLM}. 
The CLM features force sensors in the translator that measure $F(t)$, i.e., the driving force $F_y(t)$ and the out--of--plane force $F_y(t)$ and torque $T_z(t)$. 

\textbf{Data sets:} the data sets $Z_i^l$ \eqref{eq:DataSets} are generated by operating the CLM in closed--loop with the classical commutation law~\eqref{eq:Commutation_Solution_Ideal_SingleCoilset} with the initial parameters listed in Table~\ref{tab:Parameters_CLM}. 
We choose $\Delta^l(t) = (-1)^i \frac{\pi}{4}$, $i = 1,2$, to generate forces in out--of--plane directions, and thereby enable identification of the electromagnetic part in $x$ and $z$ direction. 
We supply a reference $y^*(t)$ consisting of several third order trajectories moving from $-0.1$ to $0.1$ $m$, see Fig.~\ref{fig:Data_Generation}. 
The maximum velocity is chosen as $0.025$, $0.075$ and $0.15$ $\frac{m}{s}$, and the acceleration and jerk are restricted to $1$~$\frac{m}{s^2}$ and~$1000$ $\frac{m}{s^3}$, respectively. 
For simplicity, we do not yet use a feedforward controller, i.e.,~$u_{\textup{ff}}(t) = 0$.

\begin{figure}
\centering
\includegraphics[width=1.0\linewidth]{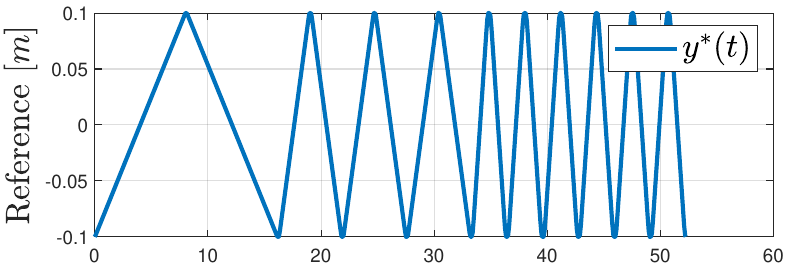}
\caption{The reference $y^*(t)$ used for generating the data sets $Z_i^l$.}
\label{fig:Data_Generation}
\end{figure}

\textbf{Commutation:} we compare the two types of commutation:
\begin{enumerate}
	\item \textbf{Classical:} the classical commutation~\eqref{eq:Commutation_Solution_Ideal_SingleCoilset} with motor constants $\hat{k}^l$ and commutation offsets $\hat{\zeta}^l$ in Table~\ref{tab:Parameters_CLM}, i.e., calibrated according to Algorithm~\ref{alg:Calibration};
	\item \textbf{PGNN:} the commutation law~\eqref{eq:Commutation_Model_Currents} with the PGNN--based model for the electromagnetic part \eqref{eq:PGNN_Electromagnetics_SingleCoilset},~\eqref{eq:GainMatrix_Cogging} identified according to Algorithm~\ref{alg:Identification} with cost function~\eqref{eq:Identification_CostFunction} using $\Lambda = 0.1 I$ and $\theta_{\textup{phy}}^*$ from the physics--based model, i.e., the PGNN without the NNs identified according to Algorithm~\ref{alg:Identification}. 
	The linear parameters are initialized as in Proposition~\ref{prop:Initialization}. The NNs have a single hidden layer $I=1$ with $n_1 = 2$ neurons in $\hat{K} \big( \theta^l, y(t) \big)$ and $n_1 = 16$ neurons in $\hat{F}_{\textup{cog}} \big( \phi^l, y(t) \big)$. 
\end{enumerate}
The cogging force identified by the PGNN is presented in Fig.~\ref{fig:Cogging_Force}.  

\begin{figure}
\centering
\includegraphics[width=1.0\linewidth]{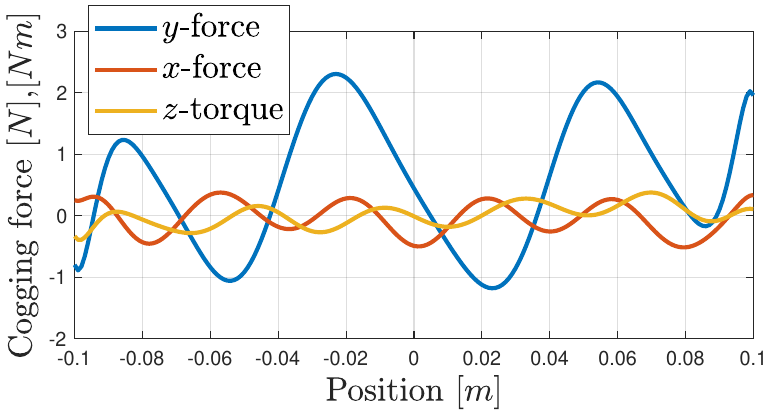}
\caption{Identified cogging force $\frac{1}{L} \sum_l \hat{F}_{\textup{cog}} \big( \hat{\phi}^l, y(t) \big)$ in~\eqref{eq:GainMatrix_Cogging} for positions $y(t)$.}
\label{fig:Cogging_Force}
\end{figure}

\textbf{Results:} following industrial usage of the CLM, we adopt the following feedforward controller 
\begin{equation}
\label{eq:FeedforwardController}
	u_{\textup{ff}} (t) = m \delta \frac{d^2}{dt^2} r(t) + f_v \delta \frac{d}{dt} r(t) + f_c \delta \textup{sign} \big( \frac{d}{dt} r(t) \big).
\end{equation}
In~\eqref{eq:FeedforwardController},~$\delta = \frac{z + 1}{2}$ compensates the delay induced by the zero--order hold discrete--to--continuous operator with $z$ the forward shift operator, and $m$, $f_v$ and $f_c$ are the mass, the viscous friction coefficient and Coulomb friction coefficient, respectively. 
The feedforward parameters $m$, $f_v$ and $f_c$ are identified using a data set obtained by operating the CLM in closed--loop with the considered commutation and the reference $y^*(t)$ in Fig.~\ref{fig:Data_Generation}. 
As a result, the feedforward parameters used for both commutation methods differ. 

\begin{figure}
\centering
\includegraphics[width=1.0\linewidth]{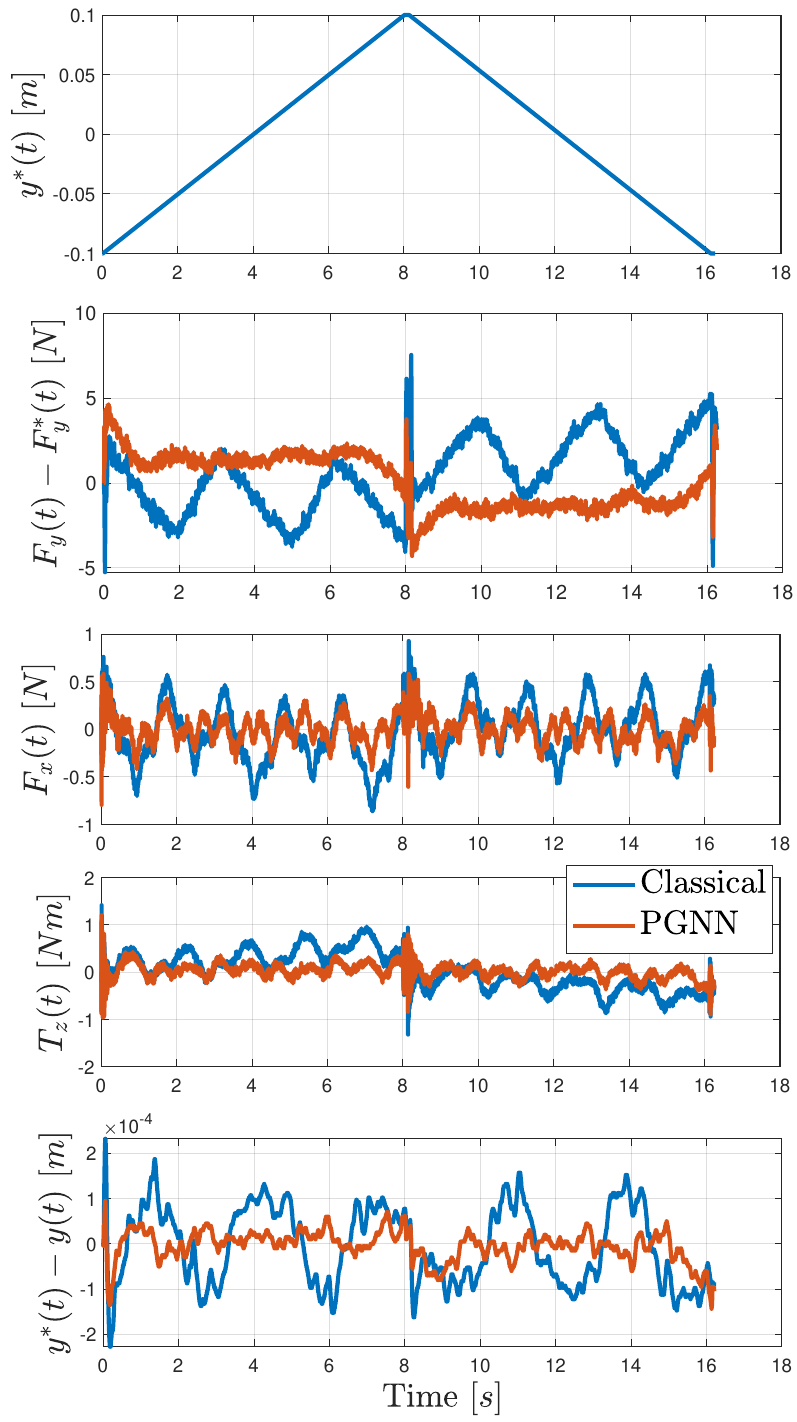}
\caption{Experimental results for different commutation on the slow reference $y^*(t)$ with maximum velocity $0.025$ $\frac{m}{s}$, from top to bottom: the reference $y^*(t)$, the commutation errors $F(t) - F^*(t)$ in $y$, $x$ and $z$ direction, and the tracking error $y^*(t) - y(t)$.}
\label{fig:CommutationError_Ref1}
\end{figure}
\begin{figure}
\centering
\includegraphics[width=1.0\linewidth]{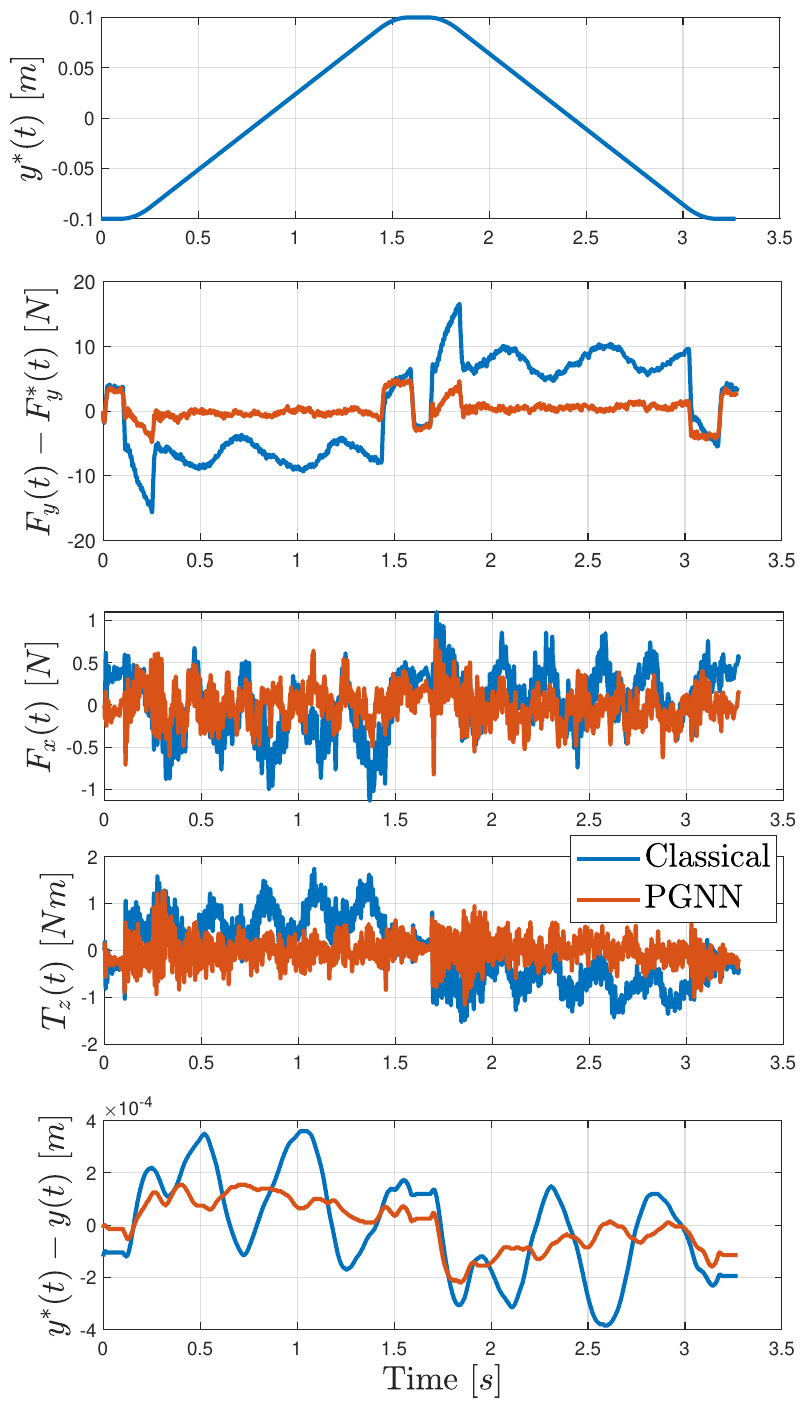}
\caption{Experimental results for different commutation on the fast reference $y^*(t)$ with maximum velocity $0.15$ $\frac{m}{s}$, from top to bottom: the reference $y^*(t)$, the commutation errors $F(t) - F^*(t)$ in $y$, $x$ and $z$ direction, and the tracking error $y^*(t) - y(t)$.}
\label{fig:CommutationError_Ref2}
\end{figure}

The filtered commutation error and resulting tracking error for the classical and PGNN commutation are visualized in Fig.~\ref{fig:CommutationError_Ref1} for a reference with maximum velocity $0.025$~$\frac{m}{s}$ and in Fig.~\ref{fig:CommutationError_Ref2} for a reference with maximum velocity $0.15$~$\frac{m}{s}$. 
The commutation error is filtered as $H(z) \big( F(t) - F^*(t) \big)$, with $H(z)$ the zero--order hold discretization of $H(s) = \frac{1}{\frac{1}{2 \pi f_{\textup{lp}}} s + 1}$ and $f_{\textup{lp}} = 50$ $Hz$.
The low--pass filter removes part of the high frequent noise present in the force measurements. 
For both references, the PGNN significantly outperforms the classical commutation, predominantly in driving direction. 
Table~\ref{tab:Commutation_Performance} gives the mean--squared error (MSE) of the filtered (and not filtered) commutation error resulting from the long reference trajectory $y^*(t)$ in Fig.~\ref{fig:Data_Generation}.
This confirms the superior performance of the PGNN commutation. 

Typically, the objective is to minimize the position tracking error $y^*(t) - y(t)$ rather than minimizing the commutation error $F(t) - F^*(t)$. 
Therefore, Figs.~\ref{fig:CommutationError_Ref1} and~\ref{fig:CommutationError_Ref2} include the tracking error resulting from the closed--loop experiments with different commutation. 
The PGNN commutation significantly outperforms the classical commutation for both references. 
For the long reference $y^*(t)$ in Fig.~\ref{fig:Data_Generation}, the MSE of the position tracking error was $2.75 \cdot 10^{-8}$~$m^2$ for classical commutation, and $5.94 \cdot 10^{-9}$~$m^2$ for PGNN commutation.

\begin{remark}
	Fig.~\ref{fig:CommutationError_Ref1} indicates a direction dependency of the commutation error $F_y(t) - F_y^*(t)$. 
	This direction dependency can be largely identified by adding $C^l \textup{sign} \big( \frac{d}{dt} y(t) \big)$, $C^l \in \mathbb{R}^3$, to $\hat{F}_{\textup{cog}}$ in~\eqref{eq:GainMatrix_Cogging}, e.g., the identification cost function~\eqref{eq:Identification_CostFunction} in $y$ direction drops from $\frac{1}{2N} \sum_{t \in Z_{1,2}^l} \Big( F_y(t) - \hat{F}_y \big( \hat{\theta}^l, \hat{\phi}^l, y(t) \big) \Big)^2 = 3.52$ to $1.20$ $N^2$ for coil set $l=1$. 
	However, the resulting commutation algorithm~\eqref{eq:Commutation_Model_Currents} becomes dynamic, or requires measurements of $\frac{d}{dt} y(t)$ which are generally not available. 
	Moreover, issues can arise in the limiting case where $\frac{d}{dt} y(t) \approx 0$, especially when overestimating $\hat{C}^l$. 
	Therefore, compensating this direction dependency is done in the feedforward controller.
\end{remark}

\begin{table}
\centering
\caption{MSE of the filtered (not filtered) commutation error achieved by original (uncalibrated) commutation, classical commutation and PGNN commutation on the reference $y^*(t)$ in Fig.~\ref{fig:Data_Generation}.}
\label{tab:Commutation_Performance}
\begin{tabular}{c || c | c | c}
	 & $\mathbf{y}$ $[N^2]$ & $\mathbf{x}$ $[N^2]$ & $\mathbf{z}$ $[N^2m^2]$ \\ \hline \hline 
	(Original) & $52.8$ ($53.4$) & $0.111$ ($0.234$) & $0.327$ ($0.781$) \\ \hline
	{Classical} & $27.4$ ($28.1$) & $0.141$ ($0.307$) & $0.436$ ($1.13$) \\ \hline
	{PGNN} & $2.45$ ($3.08$) & $0.056$ ($0.223$) & $0.180$ ($0.930$) \\ \hline
\end{tabular}
\end{table}

\section{Conclusions}
\label{sec:Conclusions}
In this work, we developed a data--driven approach to electromagnetic commutation using force sensors to minimize the commutation error, i.e., the difference between the desired and the achieved force.
A structured PGNN model class has been proposed for identification of the electromagnetic part to retain physical interpretability, while enhancing the flexibility of the model to learn parasitic effects.
Afterwards, the PGNN commutation is derived by analytically inverting the identified PGNN model while minimizing energy consumption. 
Practical implementability is ensured by deriving an input transformation for the case when the currents cannot be described directly. 
The developed approach is validated on a real--life industrial coreless linear motor, where it demonstrates significant improvements with respect to conventional commutation. 
The improvements are also observed in the reduced reference position tracking error. 

Future research will focus on merging the commutation algorithm with the feedforward control design in the position control loop, as well as the use of position measurements to remove the need for force measurements for commutation in driving direction.

\bibliographystyle{elsarticle-num} 
\bibliography{References}

\begin{thebibliography}{10}
\expandafter\ifx\csname url\endcsname\relax
  \def\url#1{\texttt{#1}}\fi
\expandafter\ifx\csname urlprefix\endcsname\relax\def\urlprefix{URL }\fi
\expandafter\ifx\csname href\endcsname\relax
  \def\href#1#2{#2} \def\path#1{#1}\fi

\bibitem{Schmidt2014}
R.~M. Schmidt, G.~Schitter, A.~Rankers, J.~van Eijk, The design of high
  performance mechatronics, {IOS} Press, 2014.

\bibitem{Nguyen2018}
T.~T. Nguyen, Identification and compensation of parasitic effects in coreless
  linear motors, Ph.D. thesis, Eindhoven University of Technology, The
  Netherlands (2018).

\bibitem{Samuelsson2005}
P.~Samuelsson, H.~Norlander, B.~Carlsson, An integrating linearization method
  for {H}ammerstein models, Automatica 41~(10) (2005) 1825--1828.

\bibitem{Johansen2013}
T.~A. Johansen, T.~I. Fosse, Control allocation -- a survey, Automatica 49~(5)
  (2013) 1087--1103.

\bibitem{Rohrig2003}
C.~R\"ohrig, Current waveform optimization for force ripple compensation of
  linear synchronous motors, {IEEE} Conference on Decision and Control (2003)
  5891--5896.

\bibitem{Moehle2015}
N.~Moehle, S.~Boyd, Optimal current waveforms for brushless permanent magnet
  motors, International Journal of Control 8~(7) (2015) 1389--1399.

\bibitem{Nguyen2016}
T.~T. Nguyen, H.~Butler, M.~Lazar, An analytical commutation law for parasitic
  forces and torques compensation in coreless linear motors, European Control
  Conference (2016) 2386--2391.

\bibitem{Custers2019}
C.~Custers, I.~Proimadis, J.~Jansen, H.~Butler, R.~T\'oth, E.~Lomonova,
  P.~van~den Hof, Active compensation of the deformation of a magnetically
  levitated mover of a planar motor, {IEEE} International Electric Machines \&
  Drives Conference (2019) 854--861.

\bibitem{Gieras2011}
J.~F. Gieras, Z.~J. Piech, B.~Tomczuk, Linear synchronous motors:
  transportation and automation systems, 2nd Edition, CRC Press, Boca Raton,
  2011.

\bibitem{Overboom2015}
T.~T. Overboom, Electromagnetic levitation and propulsion: force and torque
  decoupling in a planar motor with magnetic suspension and fail--safety, Ph.D.
  thesis, Eindhoven University of Technology, The Netherlands (2015).

\bibitem{Broens2023}
Y.~Broens, H.~Butler, R.~T\'oth, On improved commutation for moving--magnet
  planar actuators, {IEEE} Control Systems Letters 7 (2023) 2593--2598.

\bibitem{Strijbosch2019}
N.~Strijbosch, P.~Tacx, E.~Verschueren, T.~Oomen, Commutation angle iterative
  learning control: Enhancing piezo--stepper actuator waveforms, {IFAC}
  PapersOnline 52~(15) (2019) 579--584.

\bibitem{Aarnoudse2023}
L.~Aarnoudse, N.~Strijbosch, P.~Tacx, E.~Verschueren, T.~Oomen, Compensating
  commutation--angle domain disturbances with application to waveform
  optimization of piezo stepper actuators, Mechatronics 94 (2023) 103016.

\bibitem{Bascetta2010}
L.~Bascetta, P.~Rocco, G.~Magnani, Force ripple compensation in linear motors
  based on closed--loop position--dependent identification, {IEEE} Transactions
  on Mechatronics 15~(3) (2010) 349--359.

\bibitem{Rohrig2008}
C.~R{\"o}hrig, Force ripple compensation of linear synchronous motors, Asian
  Journal of Control 7~(1) (2008) 1--11.

\bibitem{Meer2022}
M.~van Meer, G.~Witvoet, T.~Oomen, Optimal commutation for switched reluctance
  motors using {G}aussian process regression, {IFAC}--PapersOnline 55~(37)
  (2022) 302--307.

\bibitem{Bolderman2021}
M.~Bolderman, M.~Lazar, H.~Butler, Physics--guided neural networks for
  inversion--based feedforward control applied to linear motors, {IEEE}
  Conference on Control Technology and Applications (2021) 1115--1120.

\bibitem{Bolderman2024}
M.~Bolderman, H.~Butler, S.~Koekebakker, E.~van Horssen, R.~Kamidin,
  T.~Spaan-Burke, N.~Strijbosch, M.~Lazar, Physics--guided neural networks for
  feedforward control with input--to--state stability guarantees, Control
  Engineering Practice 145 (2024) 105851.

\bibitem{Rovers2002}
A.~F. Rovers, Controlling the {H}--drive, Tech. rep., Eindhoven University of
  Technology, The Netherlands (2002).

\end{thebibliography}

\end{document}